\documentclass[12pt,letterpaper]{article}
\usepackage{amsmath,amsfonts,amsthm,amssymb,stmaryrd,bm,cite} 
\usepackage{relsize,tabls,tikz} 
\usepackage{graphicx} 
\usepackage{array}   

\allowdisplaybreaks  

\theoremstyle{plain}

\numberwithin{equation}{section}
\newtheorem{thm}{Theorem}[section]
\newtheorem{lem}[thm]{Lemma}
\newtheorem{cor}[thm]{Corollary}

\newcommand{\integers}{{\mathbb Z}}
\newcommand{\real}{{\mathbb R}}
\newcommand{\positive}{{\mathbb N}}
\newcommand{\complex}{{\mathbb C}}
\newcommand{\cscript}{{\mathcal C}}

\newcommand{\gscript}{{\mathcal G}}
\newcommand{\hscript}{{\mathcal H}}
\newcommand{\rscript}{{\mathcal R}}
\newcommand{\sscript}{{\mathcal S}}
\newcommand{\uscript}{{\mathcal U}}
\newcommand{\rmre}{\mathrm{Re}}
\newcommand{\rmim}{\mathop{Im}}
\newcommand{\ehat}{\widehat{e}}
\newcommand{\fhat}{\widehat{f}}
\newcommand{\ghat}{\widehat{g}}
\newcommand{\hhat}{\widehat{h}}
\newcommand{\phat}{\widehat{p}}
\newcommand{\sscripthat}{\widehat{\sscript}}
\newcommand{\capahat}{\widehat{A}}
\newcommand{\zerohat}{\widehat{0}}
\newcommand{\overa}{\overline{a}}
\newcommand{\overb}{\overline{b}}
\newcommand{\overs}{\overline{\sscript}}
\newcommand{\undere}{\underline{e}}
\newcommand{\underp}{\underline{p}}
\newcommand{\underu}{\underline{u}}
\newcommand{\underv}{\underline{v}}
\newcommand{\underw}{\underline{w}}
\newcommand{\underzero}{\underline{0}}
\newcommand{\ctimes}{\mathrel{\mathlarger\cdot}}
\newcommand*\circleds{\tikz[baseline=(char.base)]{
  \node[shape=circle,draw,inner sep=1pt] (char) {s};}}   
  
\newcommand{\ab}[1]{\left|#1\right|}
\newcommand{\doubleab}[1]{\left|\left|#1\right|\right|}
\newcommand{\brac}[1]{\left\{#1\right\}}
\newcommand{\paren}[1]{\left(#1\right)}
\newcommand{\sqbrac}[1]{\left[#1\right]}
\newcommand{\elbows}[1]{{\left\langle#1\right\rangle}}
\newcommand{\ket}[1]{{\left|#1\right>}}
\newcommand{\bra}[1]{{\left<#1\right|}}

\newcolumntype{L}{>{$}l<{$}} 

\begin{document}
\title{DISCRETE QUANTUM GRAVITY\\and\\QUANTUM FIELD THEORY
}
\author{S. Gudder\\ Department of Mathematics\\
University of Denver\\ Denver, Colorado 80208, U.S.A.\\
sgudder@du.edu
}
\date{}
\maketitle

\begin{abstract}
We introduce a discrete 4-dimensional module over the integers that appears to have maximal symmetry. By adjoining the usual Minkowski distance, we obtain a discrete 4-dimensional Minkowski space. Forming universe histories in this space and employing the standard causal order, the histories become causal sets. These causal sets increase in size rapidly and describe an inflationary period for the early universe. We next consider the symmetry group $G$ for the module. We show that $G$ has order 24 and we construct its group table. In a sense $G$ is a discrete approximation to the Lorentz group. However, we note that it contains no boosts and is essentially a rotation group. Unitary representations of $G$ are constructed. The energy-momentum space dual to the discrete module is obtained and a quantum formalism is derived. A discrete Fock space is introduced on this structure and free quantum fields are considered. Finally, we take the first step in a study of interacting quantum fields.
\end{abstract}

\section{Introduction}  
It is well-known that general relativity and quantum field theory are both plagued by singularities and infinities. This is particularly serious for quantum field theory because small distances must be considered there. These difficulties are usually circumvented by methods of infinity cancellations and re-normalizations but the methods are mathematically suspect. In general relativity, it is usually just admitted that the theory breaks down and is no longer applicable at small distance scales. The simplest and possibly only solution to these problems has been known for almost a hundred years \cite{gud68,hei30,rus54}. At that time, it was suggested by several researchers that the physical universe is discrete. They speculated that there exists an elementary length $\ell$ and an elementary time $t$. The likely values of these are the Planck length
$\ell\approx 10^{-33}$cm and the Planck time $t\approx 10^{-43}$sec. In this view, spacetime is composed of discrete, tiny cells of Planck size. This idea was not completely outrageous because it was already known that energies were composed of packets which came in multiples of Planck's constant $h$ and electric charge came in multiples of the electron charge $e$. (The latter was later altered when quarks were discovered with charges $\pm 1/3$ and $\pm 2/3$, but the idea is still the same.) One problem is that, unlike energy and charge, this granular structure of spacetime has not been experimentally observed. However, some investigators believe that with ever more sensitive instruments, this discrete framework will eventually be unveiled.

The main problem with discrete spacetime is that a substantial amount of symmetry would be lost. In particular, we would lose Lorentz invariance. Lorentz invariance is not only a pillar of theoretical physics, it has been experimentally verified a myriad of times. But it is possible that Lorentz invariance is only an approximation. There may be a smaller, more fundamental symmetry group that is indistinguishable from the Lorentz group except at very small scales.

This article attempts to construct a reasonable discrete spacetime upon which a discrete quantum gravity and quantum field theory can be built. How should one proceed with such a construction? Let's start with 2-dimensional space and add time later. We could naively begin with a square lattice structure. Basic cells that we call \textit{vertices} would be placed at $\zerohat =(0,0)$, $\ehat _1=(1,0)$, $\ehat _2=(0,1)$ and other vertices would have locations $a\ehat _1+b\ehat _2$, $a,b\in\integers$. In this case, we set Planck's length $\ell =1$. Now $\zerohat$ is indeed a distance 1 from its four nearest neighbors $\pm\ehat _1,\pm\ehat _2$. But the distance from $\zerohat$ to neighbor $\ehat _1+\ehat _2$ is
$\sqrt{2}$ and this is not a measurable distance which would have to be a positive integer. Also, $\ehat _1$ and $\ehat _2$ are distance
$\sqrt{2}$ apart. The symmetry group for this square lattice is the group of order 4 consisting of the identity and the three rotations by radian angles $\pi /2$, $\pi$ and $3\pi /2$.

We obtain more symmetry by considering a triangular lattice. In this case we have the basic vertices at $\zerohat =(0,0)$, $\ehat =(1,0)$, $\fhat =(1/2,\sqrt{3}/2)$ and other vertices have locations at $a\ehat +b\fhat$, $a,b\in\integers$. Now $\zerohat$ is a distance 1 from its six nearest neighbors $\pm\ehat ,\pm\fhat, \ehat -\fhat ,\fhat -\ehat$. We also have the bonus that $\ehat$ and $\fhat$ as well as $\ehat$ and $\ehat -\fhat$ and others are distance 1 apart. Unfortunately, the distances from $\ehat$ to $\fhat -\ehat$ is $\sqrt{3}$ which is not elementary length measurable. The triangular lattice has more symmetry than the square lattice because its symmetry group has order 6 and consists of the rotations by angles $n\pi /3$, $n=0,1,\ldots ,5$.

We noted that the distances between many lattice vertices in both cases are not elementary length measurable because they are not integers. In the first case, the distance has the form
\begin{equation*}
\doubleab{a\ehat _1+b\ehat _2}_2=\sqrt{a^2+b^2}
\end{equation*}
while in the second case
\begin{equation*}
||a\ehat +b\fhat ||_2=\sqrt{a^2+b^2+ab}
\end{equation*}
However, this is not really a problem because as usual in relativity theory, we measure distance with the Minkowski metric which is in terms of the Euclidean distance squared so instead we have
\begin{align*}
\doubleab{a\ehat _1+b\ehat _2}_2^2&=a^2+b^2
\intertext{and}
||a\ehat +b\fhat ||_2^2&=a^2+b^2+ab
\end{align*}
which, of course, are integers. Notice that these two cases are the only ones with this property. Indeed, suppose that $\ghat$ and $\hhat$ are unit vectors in $\real ^2$ with inner product $\elbows{\ghat ,\hhat}=\cos\theta$, where $0<\theta\le\pi /2$. Then for $a,b\in\integers$ we have
\begin{align*}
||a\ghat +b\hhat ||_2^2&=a^2\doubleab{\ghat}_2^2+b^2||\hhat ||_2^2+2ab\elbows{\ghat ,\hhat}\\
&=a^2+b^2+2ab\cos\theta
\end{align*}

Notice that the only values of $\theta$ for which this is always an integer are $\theta =\pi /3$ or $\pi /2$.

A similar analysis holds for 3-dimensional space. In this case we have a cubic or tetrahedral lattice. As in the 2-dimensional case, we shall see that the latter has certain advantages. Moreover, these are the only two that have measurable Minkowski distances.

In Section~2 we discuss discrete 3-dimensional space. We first construct the tetrahedral lattice $\sscript _3$ and its corresponding symmetry group $\gscript (\sscript _3)$. We show that $\gscript (\sscript _3)$ has order 24 and we exhibit its group table, In Section~3 we derive 2- and
3-dimensional unitary representations for $\gscript (\sscript _3)$.

Section~4 introduces a discrete 4-dimensional module $\sscript _4\subseteq\real ^4$ that has maximal symmetry in the sense we previously discussed. Letting $\doubleab{u}_4$ be the usual Minkowski metric, the pair $(\sscript _4,\doubleab{\ctimes}_4)$ becomes a discrete
4-dimensional Minkowski space. In a sense, its symmetry group $\gscript (\sscript _4)$ is a discrete approximation to the Lorentz group. However, we note that $\gscript (\sscript _4)$ contains no boosts and is essentially a discrete rotation group. We next form universe histories
$\sscript ^{t-}$, $t=1,2,\ldots$, where $t$ represents a discrete time. Employing the standard causal order $u<v$ if $u^0<v^0$ and
$\doubleab{u-v}_4^2\ge 0$, $\sscript ^{t-}$ becomes a causal set (causet) \cite{gud13,gud151,gud152}.

The causets $\sscript ^{t-}$ increase rapidly in size and describe an inflationary period for the early universe. During this period, the universe is essentially flat and gravity does not present itself. Moreover, the causets have a special property that we call weak covariance. This property states that all paths between two fixed vertices have the same length. At the end of the inflationary period, the system experiences a phase transition and enters the multiverse period. During the multiverse period, the universe splits into parts, each with their own geometry. The various geometries then determine curvatures and gravity in the constituent universes. Moreover, the causets possess a stronger property called covariance \cite{gud13,gud151,gud152}.

The structure presented in Section~4 entails a novel phenomenon that is more fundamental than the constancy of the speed of light $c$ in a vacuum. The reason that $c$ is the upper speed limit is that $c$ is the \textit{only} speed that a particle can attain other than zero. The reason that we observe slower speeds is that we are measuring average speeds and $c$ is the instantaneous speed of a particle.

In Section~5 the corresponding energy-momentum space $\sscripthat _4$ is obtained and a quantum formalism on $\sscripthat _4$ is derived. A discrete Fock space is introduced on $\sscripthat _4$ and free quantum fields are studied. A first step toward interacting quantum fields is presented in Section~6 and it is mentioned that this framework may result in a mathematically rigorous quantum field theory.

\section{Discrete Space} 
For contrast and comparison, we begin with a discrete space formed from a cubic lattice. To discuss this lattice, let $\ehat _1=(1,0,0)$,
$\ehat _2=(0,1,0)$, $\ehat _3=(0,0,1)$ be the usual basic vertices and let $\sscript _c^3\subseteq\real ^3$ be the set of vertices
\begin{equation*}
\sscript _c^3=\brac{u\in\real ^3\colon u=n\ehat _1+p\ehat _2+q\ehat _3,\quad n,p,q\in\integers}
\end{equation*}
Notice that $\sscript _c^3$ is a module over the integers $\integers$. That is, $u,v\in\sscript _c^3$ implies $u+v\in\sscript _c^3$ and
$nu\in\sscript _c^3$ for all $n\in\integers$. Now any $u\in\sscript _c^3$ is distance 1 to its six nearest neighbors $u\pm\ehat _1$, $u\pm\ehat _2$, $u\pm\ehat _3$. We think of the vertices in $\sscript _c^3$ as cells of Planck size that may or may not be occupied by a particle and the six edges from $u$ to its nearest neighbors as directions along which particles can move.

A \textit{symmetry} on $\sscript_c^3$ is a linear isometry on $\sscript _c^3$ with determinant 1. We denote the group of symmetries on
$\sscript _c^3$ by $\gscript (\sscript _c^3)$. The matrices of elements in $\gscript (\sscript _c^3)$ are determined by how they act on the basis
$\ehat _1$, $\ehat _2$, $\ehat _3$. For example, the symmetry 
\begin{equation*}
\capahat\colon (\ehat _1,\ehat _2,\ehat _3)\to (\ehat _2,-\ehat _1,\ehat _3)
\end{equation*}
is represented by the matrix
\begin{equation*}
\capahat =\begin{bmatrix}\noalign{\smallskip}
0&-1&0\\1&0&0\\0&0&1\\\noalign{\smallskip}\end{bmatrix}
\end{equation*}
It turns out that $\gscript (\sscript _c^3)$ is a group of order 24. These matrices are unitary with determinant 1. Besides $\capahat$ given above, the elements of $\gscript (\sscript _c^3)$ have matrices given as follows:
\begin{align*}
&\widehat{B}=\begin{bmatrix}\noalign{\smallskip}
-1&0&0\\0&-1&0\\0&0&1\\\noalign{\smallskip}\end{bmatrix},\ 
\widehat{C}=\begin{bmatrix}\noalign{\smallskip}
0&1&0\\-1&0&0\\0&0&1\\\noalign{\smallskip}\end{bmatrix},\ 
\widehat{D}=\begin{bmatrix}\noalign{\smallskip}
1&0&0\\0&0&-1\\0&1&0\\\noalign{\smallskip}\end{bmatrix},\ 
\widehat{E}=\begin{bmatrix}\noalign{\smallskip}
1&0&0\\0&-1&0\\0&0&-1\\\noalign{\smallskip}\end{bmatrix}\\\noalign{\medskip}
&\widehat{F}=\begin{bmatrix}\noalign{\smallskip}
1&0&0\\0&0&1\\0&-1&0\\\noalign{\smallskip}\end{bmatrix},\ 
\widehat{G}=\begin{bmatrix}\noalign{\smallskip}
0&0&1\\0&1&0\\-1&0&0\\\noalign{\smallskip}\end{bmatrix},\ 
\widehat{H}=\begin{bmatrix}\noalign{\smallskip}
-1&0&0\\0&1&0\\0&0&-1\\\noalign{\smallskip}\end{bmatrix},\ 
\widehat{I}=\begin{bmatrix}\noalign{\smallskip}
1&0&0\\0&1&0\\0&0&1\\\noalign{\smallskip}\end{bmatrix}\\\noalign{\medskip}
&\widehat{J}=\begin{bmatrix}\noalign{\smallskip}
0&0&-1\\0&1&0\\1&0&0\\\noalign{\smallskip}\end{bmatrix},\ 
\widehat{K}=\begin{bmatrix}\noalign{\smallskip}
0&0&1\\1&0&0\\0&1&0\\\noalign{\smallskip}\end{bmatrix},\ 
\widehat{L}=\begin{bmatrix}\noalign{\smallskip}
0&1&0\\0&0&1\\1&0&0\\\noalign{\smallskip}\end{bmatrix},\ 
\widehat{M}=\begin{bmatrix}\noalign{\smallskip}
0&-1&0\\0&0&-1\\1&0&0\\\noalign{\smallskip}\end{bmatrix}\\\noalign{\medskip}
&\widehat{N}=\begin{bmatrix}\noalign{\smallskip}
0&0&1\\-1&0&0\\0&-1&0\\\noalign{\smallskip}\end{bmatrix},\ 
\widehat{O}=\begin{bmatrix}\noalign{\smallskip}
0&-1&0\\0&0&1\\-1&0&0\\\noalign{\smallskip}\end{bmatrix},\ 
\widehat{P}=\begin{bmatrix}\noalign{\smallskip}
0&0&-1\\-1&0&0\\0&1&0\\\noalign{\smallskip}\end{bmatrix},\ 
\widehat{Q}=\begin{bmatrix}\noalign{\smallskip}
0&1&0\\0&0&1\\1&0&0\\\noalign{\smallskip}\end{bmatrix}\\\noalign{\medskip}
&\widehat{R}=\begin{bmatrix}\noalign{\smallskip}
0&1&0\\0&0&-1\\-1&0&0\\\noalign{\smallskip}\end{bmatrix},\ 
\widehat{S}=\begin{bmatrix}\noalign{\smallskip}
0&1&0\\1&0&0\\0&0&-1\\\noalign{\smallskip}\end{bmatrix},\ 
\widehat{T}=\begin{bmatrix}\noalign{\smallskip}
0&0&1\\0&-1&0\\1&0&0\\\noalign{\smallskip}\end{bmatrix},\ 
\widehat{U}=\begin{bmatrix}\noalign{\smallskip}
0&-1&0\\-1&0&0\\0&0&-1\\\noalign{\smallskip}\end{bmatrix}\\\noalign{\medskip}
&\widehat{V}=\begin{bmatrix}\noalign{\smallskip}
0&0&-1\\0&-1&0\\-1&0&0\\\noalign{\smallskip}\end{bmatrix},\ 
\widehat{W}=\begin{bmatrix}\noalign{\smallskip}
0&0&-1\\1&0&0\\0&-1&0\\\noalign{\smallskip}\end{bmatrix},\ 
\widehat{X}=\begin{bmatrix}\noalign{\smallskip}
-1&0&0\\0&0&1\\0&1&0\\\noalign{\smallskip}\end{bmatrix}\\
\end{align*}

We shall not write down the group table for $\gscript (\sscript _c^3)$ because we want to proceed to the tetrahedral space lattice which we believe has certain advantages. Let $e,f,g$ be unit vectors in $\real ^3$ that are distance 1 from each other. That is,
\begin{equation*}
\doubleab{e}=\doubleab{f}=\doubleab{g}=\doubleab{e-f}=\doubleab{e-g}=\doubleab{f-g}=1
\end{equation*}
An example of such vectors is
\begin{equation}         
\label{eq21}
e=(1,0,0),\ f=(1/2,\sqrt{3}/2,0),\ g=(1/2,1/2\sqrt{3},\sqrt{2/3})
\end{equation}
but the particular form is not needed now. The four vectors $(0,e,f,g)$ are distance 1 from each other and form the vertices of a tetrahedron with edge length 1. Notice that
\begin{equation*}
\elbows{e,f}=\elbows{e,g}=\elbows{f,g}=1/2
\end{equation*}
so the angle between any two of $e,f,g$ is $\pi /3$. As mentioned in Section~1, the cubic and tetrahedral space lattices are the only regular lattices for which the distance squared between any two vertices is an integer. The vertices of this lattice is given by the set
\begin{equation*}
\sscript _3=\brac{u\in\real ^3\colon u=ne+pf+qg,\quad n,p,q\in\integers}
\end{equation*}
As before $\sscript _3$ is a module over $\integers$. One advantage of $\sscript _3$ over $\sscript _c^3$ is that each $u\in\sscript _3$ has 12 nearest neighbors a distance 1 away so we now have 12 directions along which a particle can propagate.

\begin{lem}       
\label{lem21}
{\rm (i)}\enspace The vectors $e,f,g$ are linearly independent and form a basis for $\sscript _3$.
{\rm (i)}\enspace There are 12 unit vectors in $\sscript _3$ and these are given by $e,f,g,e-f,e-g,f-g$ and their negatives.
\end{lem}
\begin{proof}
(i)\enspace If $ne+pf+qg=0$, taking inner products with $e$, $f$ and $g$ give
\begin{equation*}
n+\tfrac{1}{2}p+\tfrac{1}{2}q=\tfrac{1}{2}n+p+\tfrac{1}{2}q=\tfrac{1}{2}n+\tfrac{1}{2}p+q=0
\end{equation*}
Solving these equations simultaneously we have $n=p=q=0$ so $e,f,g$ are linear independent. Since $\dim\real ^3=3$, $\brac{e,f,g}$ forms a basis for $\real ^3$ and hence a basis for $\sscript _3$.
(ii)\enspace If $u=ne+pf+qg$, $n,p,q\in\integers$, is a unit vector, we have that
\begin{equation*}
1=\doubleab{u}^2=n^2+p^2+q^2+np+nq+pq=\tfrac{1}{2}\sqbrac{(n+p)^2+(n+q)^2+(p+q)^2}
\end{equation*}
Hence,
\begin{equation}         
\label{eq22}
(n+p)^2+(n+q)^2+(p+q)^2=2
\end{equation}
The only way that \eqref{eq22} can hold is if two of the terms on the left side of \eqref{eq22} are 1 and the other is 0. This is possible if and only if two of the numbers $n,p,q$ are 0 and the third is $\pm 1$ or if one of the numbers is 1, another is $-1$ and the third is 0. This gives the 12 possibilities listed above.
\end{proof}

A \textit{triad} is three unit vectors in $\sscript _3$ whose inner products are $1/2$. There are 8 triads given by:
$\brac{e,f,g}$, $\brac{-e,-f,-g}$, $\brac{e,e-f,e-g}$, $\brac{-e,f-e,g-e}$, $\brac{f,f-e,f-g}$, $\brac{-f,e-f,g-f}$, $\brac{g,g-e,g-f}$,
$\quad\brac{-g,e-g,f-g}$. Each triad corresponds to three \textit{triples} written in cyclic order. For example, corresponding to triad $\brac{e,f,g}$ we have the triples $(e,f,g),(f,g,e),(g,e,f)$. We call $(e,f,g)$ the \textit{basic triple}. A \textit{symmetry} on $\sscript _3$ is a linear transformation
$T\colon\sscript _3\to\sscript _3$ that takes triples to triples and has determinant 1. The symmetries are determined by their action on the basic triple. For example if $Te=f$, $Tf=f-e$ and $Tg=f-g$, then $T$ preserves all triples. For instance,
\begin{align*}
T(-g)&=g-f\\
T(e-g)&=Te-Tg=f-f+g=g\\
T(f-g)&=Tf-Tg=f-e-f+g=g-e
\end{align*}
Hence, $T(-g,e-g,f-g)=(g-f,g,g-e)$.

Since there are 24 triples, we conclude that there are 24 symmetries. The simplest are the identity $I(e,f,g)=(e,f,g)$ and $A(e,f,g)=(f,g,e)$. We can write the symmetries relative to the basis $\brac{e,f,g}$ as follows:
\begin{align*}
&\hskip -10pt A=\begin{bmatrix}\noalign{\smallskip}
0&0&1\\1&0&0\\0&1&0\\\noalign{\smallskip}\end{bmatrix},\ 
B=\begin{bmatrix}\noalign{\smallskip}
0&1&0\\0&0&1\\1&0&0\\\noalign{\smallskip}\end{bmatrix},\ 
C=\begin{bmatrix}\noalign{\smallskip}
-1&-1&-1\\0&1&0\\1&0&0\\\noalign{\smallskip}\end{bmatrix},\ 
D=\begin{bmatrix}\noalign{\smallskip}
0&0&1\\0&1&0\\-1&-1&-1\\\noalign{\smallskip}\end{bmatrix}\\\noalign{\medskip}
&\hskip -10pt E=\begin{bmatrix}\noalign{\smallskip}
1&0&0\\-1&-1&-1\\0&1&0\\\noalign{\smallskip}\end{bmatrix},\ 
F=\begin{bmatrix}\noalign{\smallskip}
1&0&0\\0&0&1\\-1&-1&-1\\\noalign{\smallskip}\end{bmatrix},\ 
G=\begin{bmatrix}\noalign{\smallskip}
-1&-1&-1\\1&0&0\\0&0&1\\\noalign{\smallskip}\end{bmatrix},\ 
H=\begin{bmatrix}\noalign{\smallskip}
0&1&0\\-1&-1&-1\\0&0&1\\\noalign{\smallskip}\end{bmatrix}\\\noalign{\medskip}
&\hskip -10pt I=\begin{bmatrix}\noalign{\smallskip}
1&0&0\\0&1&0\\0&0&1\\\noalign{\smallskip}\end{bmatrix},\ 
J=\begin{bmatrix}\noalign{\smallskip}
-1&-1&-1\\0&0&1\\0&1&0\\\noalign{\smallskip}\end{bmatrix},\ 
K=\begin{bmatrix}\noalign{\smallskip}
0&1&0\\1&0&0\\-1&-1&-1\\\noalign{\smallskip}\end{bmatrix},\ 
L=\begin{bmatrix}\noalign{\smallskip}
0&0&1\\-1&-1&-1\\1&0&0\\\noalign{\smallskip}\end{bmatrix}\\\noalign{\medskip}
&\hskip -10pt M=\begin{bmatrix}\noalign{\smallskip}
1&-1&1\\0&-1&0\\0&0&-1\\\noalign{\smallskip}\end{bmatrix},\ 
N=\begin{bmatrix}\noalign{\smallskip}
0&0&-1\\1&1&1\\0&-1&0\\\noalign{\smallskip}\end{bmatrix},\ 
O=\begin{bmatrix}\noalign{\smallskip}
0&-1&0\\0&0&-1\\1&1&1\\\noalign{\smallskip}\end{bmatrix},\ 
P=\begin{bmatrix}\noalign{\smallskip}
-1&0&0\\0&-1&0\\1&1&1\\\noalign{\smallskip}\end{bmatrix}\\\noalign{\medskip}
&\hskip -10pt Q=\begin{bmatrix}\noalign{\smallskip}
0&0&-1\\0&-1&0\\-1&0&0\\\noalign{\smallskip}\end{bmatrix},\ 
R=\begin{bmatrix}\noalign{\smallskip}
1&1&1\\-1&0&0\\0&-1&0\\\noalign{\smallskip}\end{bmatrix},\ 
S=\begin{bmatrix}\noalign{\smallskip}
1&1&1\\0&0&-1\\-1&0&0\\\noalign{\smallskip}\end{bmatrix},\ 
T=\begin{bmatrix}\noalign{\smallskip}
-1&0&0\\1&1&1\\0&0&-1\\\noalign{\smallskip}\end{bmatrix}\\\noalign{\medskip}
&\hskip -10pt U=\begin{bmatrix}\noalign{\smallskip}
0&-1&0\\-1&0&0\\0&0&-1\\\noalign{\smallskip}\end{bmatrix},\ 
V=\begin{bmatrix}\noalign{\smallskip}
-1&0&0\\0&0&-1\\0&-1&0\\\noalign{\smallskip}\end{bmatrix},\ 
W=\begin{bmatrix}\noalign{\smallskip}
0&-1&0\\1&1&1\\-1&0&0\\\noalign{\smallskip}\end{bmatrix},\ 
X=\begin{bmatrix}\noalign{\smallskip}
0&0&-1\\-1&0&0\\1&1&1\\\noalign{\smallskip}\end{bmatrix}\\
\end{align*}

We denote the group of symmetries $\brac{A,B,\ldots ,X}$ by $\gscript (\sscript _3)$ There are other linear transformations on $\sscript _3$ that take triples to triples but these do not have unit determinant. For example,
\begin{equation*}
A'=\begin{bmatrix}\noalign{\smallskip}
1&0&0\\0&0&1\\0&1&0\\\noalign{\smallskip}\end{bmatrix}
\end{equation*}
has $\det A=-1$. Even though $\gscript (\sscript _c^3)$ and $\gscript (\sscript _3)$ both have order 24, they are different groups and we believe that $\gscript (\sscript _3)$ is more interesting. The group table for $\gscript (\sscript _3)$ is given by Table~1.


Some of the subgroups of $\gscript (\sscript _3)$ are: $\brac{I,A,B}$, $\brac{I,C,D}$, $\brac{I,E,F}$, $\brac{I,G,H}$, $\brac{I,W,J,X}$,
$\brac{I,A,B,U,V,Q}$, $\brac{A,B,\ldots ,L}$, $\brac{I,J}$, $\brac{I,K}$, $\brac{I,L}$, $\brac{I,T}$, $\brac{I,U}$, $\brac{I,V}$, $\brac{I,M}$. We can check from Table~1 that $M$ and $N$ generate $\gscript (\sscript _3)$. In fact, any two non-commuting elements from $\brac{M,N\ldots ,X}$ are generators. As usual, a linear transformation $Z\colon\sscript _3\to\sscript _3$ is an \textit{isometry} if $\elbows{Zu,Zv}=\elbows{u,v}$ for all $u,v\in\sscript _3$. It is clear that an isometry with unit determinant is a symmetry. Conversely, it is easy to check that $M$ and $N$ are isometries and since $M$ and $N$ generate $\gscript (\sscript _3)$, it follows that every element of $\gscript (\sscript _3)$ is an isometry. We conclude that our definition of a symmetry on $\sscript _3$ coincides with our original definition of a symmetry.

{\parindent=-75pt
\begin{tabular}{|L|L|L|L|L|L|L|L|L|L|L|L|L|L|L|L|L|L|L|L|L|L|L|L|L|L|L|L|}
&I&A&B&C&D&E&F&G&H&J&K&L&M&N&O&P&Q&R&S&T&U&V&W&X\\
\hline
I&I&A&B&C&D&E&F&G&H&J&K&L&M&N&O&P&Q&R&S&T&U&V&W&X\\
\hline
A&A&B&I&E&J&K&G&L&D&H&C&F&N&O&M&R&V&W&T&X&Q&U&P&S\\
\hline
B&B&I&A&K&H&C&L&F&J&D&E&G&O&M&N&W&U&P&X&S&V&Q&R&T\\
\hline
C&C&G&J&D&I&L&B&K&E&F&A&H&P&T&V&Q&M&X&O&W&R&S&N&U\\
\hline
D&D&K&F&I&C&H&J&A&L&B&G&E&Q&W&S&M&P&U&V&N&X&O&T&R\\
\hline
E&E&L&H&J&A&F&I&C&K&G&B&D&R&X&U&V&N&S&M&P&W&T&O&Q\\
\hline
F&F&D&K&G&L&I&E&J&B&C&H&A&S&Q&W&T&X&M&R&V&O&P&U&N\\
\hline
G&G&J&C&L&F&A&K&H&I&E&D&B&T&V&P&X&S&N&W&U&M&R&Q&O\\
\hline
H&H&E&L&B&K&J&D&I&G&A&F&C&U&R&X&O&W&V&Q&M&T&N&S&P\\
\hline
J&J&C&G&A&E&D&H&B&F&I&L&K&V&P&T&N&R&Q&U&O&S&M&X&W\\
\hline
K&K&F&D&H&B&G&A&E&C&L&I&J&W&S&Q&U&O&T&N&R&P&X&M&V\\
\hline
L&L&H&E&F&G&B&C&D&A&K&J&I&X&U&R&S&T&O&P&Q&N&W&V&M\\
\hline
M&M&R&S&Q&P&N&O&U&T&V&X&W&I&E&F&D&C&A&B&H&G&J&L&K\\
\hline
N&N&W&T&V&R&O&M&Q&X&U&S&P&A&K&G&J&E&B&I&D&L&H&F&C\\
\hline
O&O&P&X&U&W&M&N&V&S&Q&T&R&B&C&L&H&K&I&A&J&F&D&G&E\\
\hline
P&P&X&O&M&Q&T&V&R&W&S&U&N&C&L&B&I&D&G&J&E&K&F&H&A\\
\hline
Q&Q&U&V&P&M&W&S&X&N&O&R&T&D&H&J&C&I&K&F&L&A&B&E&G\\
\hline
R&R&S&M&N&V&X&U&W&P&T&Q&O&E&F&I&A&J&L&H&K&C&G&D&B\\
\hline
S&S&M&R&X&T&Q&W&O&V&P&N&U&F&I&E&L&G&D&K&B&J&C&A&H\\
\hline
T&T&N&W&S&X&V&P&M&U&R&O&Q&G&A&K&F&L&J&C&I&H&E&B&D\\
\hline
U&U&V&Q&W&O&R&X&T&M&N&P&S&H&J&D&K&B&E&L&G&I&A&C&F\\
\hline
V&V&Q&U&R&N&P&T&S&O&M&W&X&J&D&H&E&A&C&G&F&B&I&K&L\\
\hline
W&W&T&N&O&U&S&Q&P&R&X&V&M&K&G&A&B&H&F&D&C&E&L&J&I\\
\hline
X&X&O&P&T&S&U&R&N&Q&W&M&V&L&B&C&G&F&H&E&A&D&K&I&J\\
\hline\noalign{\medskip}
\multicolumn{24}{c}{\textbf{Table 1}}\\
\end{tabular}
\parindent=18pt}

\section{Unitary Representations} 
This section constructs unitary representations for the group $\gscript (\sscript _3)$. The matrix realizations that we gave for $\gscript (\sscript _3)$ were not necessarily unitary because they were given relative to the nonorthogonal basis $\brac{e,f,g}$. Since these matrices are isometries, if we represent them in the standard basis $\ehat _1=(1,0,0)$, $\ehat _2=(0,1,1)$, $\ehat _3=(0,0,1)$, they will become unitary. Using the concrete form \eqref{eq21} for $e,f,g$, the two bases are related by
\begin{align*}
e&=\ehat _1,\quad f=\frac{1}{2}\,\ehat _1+\frac{\sqrt{3}}{2}\,\ehat _2,\quad
g=\frac{1}{2}\,\ehat _1+\frac{1}{2\sqrt{3}}\,\ehat _2+\sqrt{\frac{2}{3}}\,\,\ehat _3\\
\intertext{and}
\ehat _1&=e,\quad\ehat _2=-\frac{1}{\sqrt{3}}\,e+\frac{2}{\sqrt{3}}\,f,\quad \ehat _3=-\frac{1}{\sqrt{6}}\,e-\frac{1}{\sqrt{6}}\,f+\sqrt{\frac{3}{2}}\,\,g
\end{align*}
The basis transformations become:
\begin{equation*}
U=\begin{bmatrix}\noalign{\smallskip}
1&-1/\sqrt{3}&-1/\sqrt{6}\\0&2/\sqrt{3}&-1/\sqrt{6}\\0&0&\sqrt{3}/2\\\noalign{\smallskip}\end{bmatrix},\quad
U^{-1}=\begin{bmatrix}\noalign{\smallskip}
1&1/2&1/2\\0&\sqrt{3}/2&1/2\sqrt{3}\\0&0&\sqrt{2/3}\\\noalign{\smallskip}\end{bmatrix}
\end{equation*}
Define the unitary representation $\uscript$ of $\gscript (\sscript _3)$ on $\complex ^3$ by
\begin{equation*}
\uscript (Z)=U^{-1}ZU
\end{equation*}
for all $Z\in\gscript (\sscript _3)$. This is a group representation because
\begin{equation*}
\uscript (YZ)=U^{-1}YZU=U^{-1}YUU^{-1}ZU=\uscript (Y)\uscript (Z)
\end{equation*}

Since $M$ and $N$ generate $\gscript (\sscript _3)$ we find $\uscript (M)$, $\uscript (N)$ to be:
\begin{equation*}
\uscript (M)=\begin{bmatrix}\noalign{\smallskip}
1&0&0\\0&-1&0\\0&0&-1\\\noalign{\smallskip}\end{bmatrix},\quad
\uscript (N)=\begin{bmatrix}\noalign{\smallskip}
1/2&-1/2\sqrt{3}&-\sqrt{2/3}\\\sqrt{3}/2&1/6&\sqrt{2}/3\\0&-2\sqrt{2}/3&1/3\\\noalign{\smallskip}\end{bmatrix}
\end{equation*}
Clearly, $\uscript (M)$ is a unitary and it is easy to check that $\uscript (N)$ is unitary. In fact, they are orthogonal matrices because they have real entries. Since $M$ and $N$ generate $\gscript (\sscript _3)$, it again follows that $\uscript (Z)$ is unitary for all $Z\in\gscript (\sscript _3)$ so
$\uscript$ is a unitary representation. Since $\gscript (\sscript _3)$ and $\uscript\sqbrac{\gscript (\sscript _3)}$ are isomorphic groups, we can and frequently will identify them.

All the matrices in $\gscript (\sscript _3)$ have eigenvalues among the numbers $\pm 1$, $\pm i$, $e^{\pm 2\pi i/3}$ with possible multiplicities. Of course, their corresponding eigenvectors are different, in general. For example, $A$ has eigenvalues $1,e^{2\pi i/3},e^{-2\pi i/3}$ with corresponding (unnormalized) eigenvectors $(1,1,1)$, $(e^{2\pi i/3},1,e^{-2\pi i/3})$, $(e^{-2\pi i/3},1,e^{2\pi i/3})$ relative to the $\brac{e,f,g}$ basis. It is easy to check that $B,C,\ldots ,H$ have the same eigenvalues as $A$. A new pattern begins with $J$ which has eigenvalues $1,-1$ (multiplicity 2) with corresponding eigenvectors $(1,-1,-1)$, $(1,-1,1)$, $(1,0,0)$. The matrices $K,L,M,Q$ have the same eigenvalues as $J$. The matrix $M$ has eigenvalues $1,i,-i$ and corresponding eigenvectors $(1,0,0)$, $(1,-1,-i)$, $(1,-1,i)$. The other matrices are similar to those already computed.

We can apply this work to find eigenvalues and eigenvectors for the unitary matrices in $\uscript\sqbrac{\gscript (\sscript _3)}$. For example, we have found the eigenvalues $\lambda _1,\lambda _2,\lambda _3$ and corresponding eigenvectors $u_1,u_2,u_3$ for $A\in\gscript (\sscript _3)$. Now $\uscript (A)=U^{-1}AU$ so the eigenvalues of $\uscript (A)$ are $\lambda _j$, $j=1,2,3$ with corresponding eigenvectors $U^{-1}u_j$, $j=1,2,3$. In particular, the eigenvalues of $\uscript (A)$ are $1,e^{2\pi i/3},e^{-2\pi i/3}$ with corresponding eigenvectors
\begin{align*}
U^{-1}&=\begin{bmatrix}\noalign{\smallskip}1\\1\\1\\\noalign{\smallskip}\end{bmatrix}=
\begin{bmatrix}\noalign{\smallskip}2\\2\sqrt{3}\\\sqrt{2/3}\\\noalign{\smallskip}\end{bmatrix}\\\noalign{\smallskip}
U^{-1}&=\begin{bmatrix}\noalign{\smallskip}e^{2\pi i/3}\\1\\e^{-2\pi i/3}\\\noalign{\smallskip}\end{bmatrix}=\frac{1}{4}
\begin{bmatrix}\noalign{\smallskip}5+i\sqrt{3}\\5\sqrt{3}-i\\2^{3/2}(1/\sqrt{3}-i)\\\noalign{\smallskip}\end{bmatrix}\\\noalign{\smallskip}
U^{-1}&=\begin{bmatrix}\noalign{\smallskip}e^{-2\pi i/3}\\1\\e^{2\pi i/3}\\\noalign{\smallskip}\end{bmatrix}=\frac{1}{4}
\begin{bmatrix}\noalign{\smallskip}5-i\sqrt{3}\\5\sqrt{3}+i\\2^{3/2}(1/\sqrt{3}+i)\\\noalign{\smallskip}\end{bmatrix}
\end{align*}

It is useful to find the eigenvalues and eigenvectors for the unitary operators $\uscript\sqbrac{\gscript (\sscript _3)}$. This is because the self-adjoint generator of $\uscript (Z)$, $Z\in\gscript (\sscript _3)$, which gives an angular momentum operator can then be derived. For example if
$\uscript (A)=e^{iA'}$ has generator $A'$, then the eigenvalues of $A'$ are $0,-2\pi /3,2\pi /3$ and the eigenvectors of $A'$ are the same as those for $\uscript (A)$.

We now construct a unitary representation of $\uscript\sqbrac{\gscript (\sscript _3)}$ on the 2-dimensional Hilbert space $\complex ^2$. The standard construction goes as follows \cite{sw64,vel94}. For $\capahat\in\uscript\sqbrac{\gscript (\sscript _3)}$, the corresponding unitary matrix on $\complex ^2$ is given by
\begin{equation*}
\rscript '(\capahat )=\begin{bmatrix}\noalign{\smallskip}
a&b\\-\overb&\overa\\\noalign{\smallskip}\end{bmatrix}
\end{equation*}
and $a,b\in\complex$ satisfy the following seven equations. In these equations, we denote a vector $u\in\real ^3$ by its Cartesian coordinates $u=(u_1,u_2,u_3)$.
\begin{align}        
\label{eq31}    
a^2-b^2&=(\capahat\ehat _1)_1-i(\capahat\ehat _1)_2\\
\label{eq32}    
ab&=-\tfrac{1}{2}(\capahat\ehat _3)_1+\tfrac{i}{2}(\capahat\ehat _3)_2\\
\label{eq33}    
a\overb &=\tfrac{1}{2}(\capahat\ehat _1)_3+\tfrac{i}{2}(\capahat\ehat _2)_2\\
\label{eq34}    
\rmre (a^2+b^2)&=(\capahat\ehat _2)_2\\
\label{eq35}    
\rmim (b^2-\overa ^2)&=(\capahat\ehat _2)_1\\
\label{eq36}    
\ab{b}^2&=\tfrac{1}{2}\sqbrac{1-(\capahat\ehat _3)_3}\\
\label{eq37}    
\ab{a}^2+\ab{b}^2&=1
\end{align}
We define the unitary representation $\rscript$ of $\gscript (\sscript _3)$ on $\complex ^2$ by $\rscript (Z)=\rscript '\sqbrac{\uscript (X)}$,
$Z\in\gscript (\sscript _3)$. Following the above procedure we obtain
\begin{equation*}
\rscript (M)=i\begin{bmatrix}\noalign{\smallskip}0&1\\1&0\\\noalign{\smallskip}\end{bmatrix}
\end{equation*}
We also have that
\begin{equation*}
\rscript (N)=\frac{1}{\sqrt{3}}\begin{bmatrix}\noalign{\smallskip}\sqrt{2}e^{-i\pi /6}&e^{i\pi /3}\\
-e^{-i\pi /3}&\sqrt{2}e^{i\pi /6}\\\noalign{\smallskip}\end{bmatrix}
\end{equation*}

Since $M$ and $N$ generate $\gscript (\sscript _3)$ we can obtain $\rscript (Z)$ for every $Z\in\gscript (\sscript _3)$ by repeated applications of
$\rscript (M)$ and $\rscript (N)$ using Table~1. These are now listed.
\begin{align*}
\rscript (A)&=\frac{i}{\sqrt{3}}\begin{bmatrix}\noalign{\smallskip}e^{i\pi /3}&\sqrt{2}e^{-i\pi /6}\\\noalign{\smallskip}
\sqrt{2}e^{i\pi /6}&-e^{-i\pi /3}\\\noalign{\smallskip}\end{bmatrix},\quad
\rscript (B)=-\frac{1}{\sqrt{3}}\begin{bmatrix}\noalign{\smallskip}e^{i\pi /6}&\sqrt{2}e^{i\pi /3}\\\noalign{\smallskip}
-\sqrt{2}e^{-i\pi /3}&e^{-i\pi /6}\\\noalign{\smallskip}\end{bmatrix}\\\noalign{\smallskip}
\rscript (C)&=\frac{1}{\sqrt{3}}\begin{bmatrix}\noalign{\smallskip}e^{i5\pi /6}&-\sqrt{2}\\\noalign{\smallskip}
\sqrt{2}&e^{-i\pi /6}\\\noalign{\smallskip}\end{bmatrix},\quad
\rscript (D)=-\frac{1}{\sqrt{3}}\begin{bmatrix}\noalign{\smallskip}e^{i\pi /6}&-\sqrt{2}\\\noalign{\smallskip}
\sqrt{2}&e^{-i\pi /6}\\\noalign{\smallskip}\end{bmatrix}\\\noalign{\smallskip}
\rscript (E)&=\frac{i}{\sqrt{3}}\begin{bmatrix}\noalign{\smallskip}-e^{-i\pi /3}&\sqrt{2}e^{i\pi /6}\\\noalign{\smallskip}
\sqrt{2}e^{-i\pi /6}&e^{i\pi /3}\\\noalign{\smallskip}\end{bmatrix},\quad
\rscript (F)=\frac{1}{\sqrt{3}}\begin{bmatrix}\noalign{\smallskip}e^{i5\pi /6}&\sqrt{2}e^{i2\pi /3}\\\noalign{\smallskip}
-\sqrt{2}e^{-i2\pi /3}&e^{-i5\pi /6}\\\noalign{\smallskip}\end{bmatrix}\\\noalign{\smallskip}
\rscript (G)&=i\begin{bmatrix}\noalign{\smallskip}e^{i\pi /6}&0\\\noalign{\smallskip}
0&-e^{-i\pi /6}\\\noalign{\smallskip}\end{bmatrix},\quad
\rscript (H)=i\begin{bmatrix}\noalign{\smallskip}-e^{-i\pi /6}&0\\\noalign{\smallskip}
0&-e^{i\pi /6}\\\noalign{\smallskip}\end{bmatrix}\\\noalign{\smallskip}
\rscript (I)&=\begin{bmatrix}\noalign{\smallskip}1&0\\\noalign{\smallskip}0&1\\\noalign{\smallskip}\end{bmatrix},\quad
\rscript (J)=\frac{1}{\sqrt{3}}\begin{bmatrix}\noalign{\smallskip}-i&-\sqrt{2}\\\noalign{\smallskip}
\sqrt{2}&i\\\noalign{\smallskip}\end{bmatrix}\\\noalign{\smallskip}
\rscript (K)&=\frac{1}{\sqrt{3}}\begin{bmatrix}\noalign{\smallskip}e^{-i\pi /2}&\sqrt{2}e^{i\pi /3}\\\noalign{\smallskip}
-\sqrt{2}e^{-i\pi /3}&e^{i\pi /2}\\\noalign{\smallskip}\end{bmatrix},\quad
\rscript (L)=-\frac{1}{\sqrt{3}}\begin{bmatrix}\noalign{\smallskip}e^{i\pi /2}&\sqrt{2}e^{i2\pi /3}\\\noalign{\smallskip}
-\sqrt{2}e^{-i2\pi /3}&e^{-i\pi /2}\\\noalign{\smallskip}\end{bmatrix}\\\noalign{\smallskip}
\rscript (O)&=\frac{1}{\sqrt{3}}\begin{bmatrix}\noalign{\smallskip}\sqrt{2}e^{-i\pi /6}&e^{-i\pi /3}\\\noalign{\smallskip}
-e^{i\pi /3}&\sqrt{2}e^{i\pi /6}\\\noalign{\smallskip}\end{bmatrix},\quad
\rscript (P)=\frac{i}{\sqrt{3}}\begin{bmatrix}\noalign{\smallskip}\sqrt{2}&e^{-i\pi /6}\\\noalign{\smallskip}
e^{i\pi /6}&-\sqrt{2}\\\noalign{\smallskip}\end{bmatrix}\\\noalign{\smallskip}
\rscript (Q)&=\frac{i}{\sqrt{3}}\begin{bmatrix}\noalign{\smallskip}\sqrt{2}&-e^{-i\pi /6}\\\noalign{\smallskip}
-e^{i\pi /6}&-\sqrt{2}\\\noalign{\smallskip}\end{bmatrix},\quad
\rscript (R)=-\frac{1}{\sqrt{3}}\begin{bmatrix}\noalign{\smallskip}\sqrt{2}e^{i\pi /6}&-e^{-i\pi /3}\\\noalign{\smallskip}
e^{i\pi /3}&\sqrt{2}e^{-\pi /6}\\\noalign{\smallskip}\end{bmatrix}\\\noalign{\smallskip}
\rscript (S)&=\frac{1}{\sqrt{3}}\begin{bmatrix}\noalign{\smallskip}\sqrt{2}e^{i\pi /6}&-e^{i\pi /3}\\\noalign{\smallskip}
-e^{i\pi /6}&-\sqrt{2}\\\noalign{\smallskip}\end{bmatrix},\quad
\rscript (T)=\begin{bmatrix}\noalign{\smallskip}0&-e^{i\pi /6}\\\noalign{\smallskip}
e^{-i\pi /6}&0\\\noalign{\smallskip}\end{bmatrix}\\\noalign{\smallskip}
\rscript (U)&=\begin{bmatrix}\noalign{\smallskip}0&-e^{-i\pi /6}\\\noalign{\smallskip}
e^{i\pi /6}&0\\\noalign{\smallskip}\end{bmatrix},\quad
\rscript (V)=\frac{1}{\sqrt{3}}\begin{bmatrix}\noalign{\smallskip}i\sqrt{2}&1\\\noalign{\smallskip}
-1&-i\sqrt{2}\\\noalign{\smallskip}\end{bmatrix}\\\noalign{\smallskip}
\rscript (W)&=\frac{1}{\sqrt{3}}\begin{bmatrix}\noalign{\smallskip}\sqrt{2}e^{i5\pi /6}&1\\\noalign{\smallskip}
-1&\sqrt{2}e^{-i5\pi /6}\\\noalign{\smallskip}\end{bmatrix},\quad
\rscript (X)=\frac{i}{\sqrt{3}}\begin{bmatrix}\noalign{\smallskip}\sqrt{2}e^{i\pi /6}&1\\\noalign{\smallskip}
-1&\sqrt{2}e^{-i\pi /6}\\\noalign{\smallskip}\end{bmatrix}\\\noalign{\smallskip}
\end{align*}

Technically speaking, $\rscript$ is a \textit{projective representation} of $\gscript (\sscript _3)$ on $\complex ^2$. That is
$\rscript (YZ)=\pm\rscript (Y)\rscript (Z)$ in general. However, this is not important because quantum states are only determined within a scalar multiple of absolute value one. Examples are $GH=I$ but $\rscript (G)\rscript (H)=-I$ and $J^2=I$ and $\rscript (J)^2=-I$.

\section{Discrete Spacetime} 
We have previously considered the space lattice $\sscript _3$. We now adjoin time to obtain a spacetime lattice. Let $d,e,f,g\in\real$ be unit vectors satisfying
\begin{align*}
\elbows{d,e}=\elbows{d,f}=\elbows{d,g}=0
\intertext{and}
\elbows{e,f}=\elbows{e,g}=\elbows{f,g}=1/2
\end{align*}
Then
\begin{equation*}
\sscript _4=\brac{td+ne+pf+qg\colon t,n,p,q\in\integers}
\end{equation*}
is a 4-dimensional module over $\integers$ with basis $\brac{d,e,f,g}$. Clearly, $\sscript _3$ is a 3-dimensional submodule of $\sscript _4$. We are mainly concerned with the subset $\sscript _4^+\subseteq\sscript _4$ of vectors with $t\ge 0$. We frequently call these vectors \textit{vertices} and consider them to be tiny spacetime cells that may be occupied by a particle. If $u=td+ne+pf+qg$, we write $u=(t,n,p,q)$ and use the notation $u^0=t$, $u^1=n$, $u^2=p$, $u^3=q$.

We define the usual norm $\doubleab{\ctimes}_3$ on $\sscript _3$ by
\begin{equation*}
\doubleab{ne+pf+qg}_3^2=n^2+p^2+q^2+np+nq+pq
\end{equation*}
and the indefinite norm $\doubleab{\ctimes}_4$ on $\sscript _4$ by
\begin{equation*}
\doubleab{td+ne+pf+qg}_4^2=t^2-\doubleab{ne+pf+qg}_3^2
\end{equation*}
We sometimes write $\doubleab{u}_4=(u^0)^2-\doubleab{\underu}_3^2$ where $\underu =(u^1,u^2,u^3)\in\real ^3$. Of course, the
\textit{spacetime distance between} $u,v\in\sscript _4$ is $\doubleab{u-v}_4$. For $t\ge 0$, we define
\begin{equation*}
\sscript ^t=\brac{u\in\sscript _4^+\colon u^0=t,\ \doubleab{u}_4^2\ge 0}
\end{equation*}
and we call $\sscript ^t$ the \textit{universe at time} $t$. We call
\begin{equation*}
\overs  ^t=\cup\brac{\sscript ^t\colon 0\le t'\le t,t'\in\integers}
\end{equation*}
the \textit{universe history until time} $t$. For $u\in\sscript _4^+$, the \textit{forward light cone at} $u$ is
\begin{equation*}
\cscript _u^+=\brac{v\in\sscript _4^+\colon v^0\ge u^0,\ \doubleab{v-u}_4^2\ge 0}
\end{equation*}
As usual $\cscript _u^+$ is the set of vertices that $u$ can reach with a physical signal. In other words $\cscript _u^+$ is the set of vertices that $u$ can influence. Of course,
\begin{equation*}
\cscript _0^+=\brac{u\in\sscript _4\colon u^0\ge 0,\ \doubleab{u}_4\ge 0}=\bigcup _{t\ge0}\sscript ^t=\bigcup _{t\ge 0}\overs  ^t
\end{equation*}
We consider $\cscript _0^+$ to be the spacetime background and shall only consider vertices in $\cscript _0^+$. The \textit{forward null surface} at $u\in\cscript _0^+$ is 
\begin{equation*}
\eta _u^+=\brac{v\in\sscript _4\colon v^0\ge u^0,\ \doubleab{v-u}_4^2=0}
\end{equation*}
We interpret $\eta _u^+$ as the vertices that $u$ can reach with a light signal. For $u,v\in\cscript _0^+$ we write $u<v$ and say that $v$ is in the
\textit{causal future} of $u$ if $u^0<v^0$ and $\doubleab{v-u}_4^2\ge 0$. Of course, $u<v$ if and only if $v\in\cscript _u^+\smallsetminus\brac{u}$.

\begin{thm}       
\label{thm41}
The relation $<$ is a partial order on $\cscript _0^+$.
\end{thm}
\begin{proof}
Clearly $u\not <u$ for all $u\in\cscript _0^+$. Suppose $u,v,w\in\cscript _0^+$ with $u<v$ and $v<w$. Then $u^0<v^0<w^0$ and
\begin{equation}         
\label{eq41}
\begin{aligned}
(v^0-u^0)^2-\doubleab{\underv -\underu}_3^2&=\doubleab{v-u}_4^2\ge 0\\
(w^0-v^0)^2-\doubleab{\underw -\underv}_3^2&=\doubleab{u-v}_4^2\ge 0
\end{aligned}
\end{equation}
By the triangle inequality we have
\begin{align*}
\doubleab{w-u}_4^2&=(w^0-u^0)^2-\doubleab{\underw -\underu}_3^2
  \ge (w^0-u^0)^2-\doubleab{\underw -\underv}_3^2-\doubleab{\underv -\underu}_3^2\\
&=(w^0-v^0)^2+(v^0-u^0)^2+2(w^0-v^0)(v^0-u^0)-\doubleab{\underw -\underv}_3^2-\doubleab{\underv-\underu}_3^2
\end{align*}
Then $\doubleab{w-u}_4^2\ge 0$ follows from \eqref{eq41}.
\end{proof}

The partial order $<$ restricted to $\overs  ^t$ makes $\overs  ^t$ a finite poset that is frequently called a \textit{causal set} or
\textit{causet} \cite{gud13,gud151,gud152}. If $u<v$ and there is no $w\in\cscript _0^+$ with $u<w<v$, we say that $v$ is a \textit{child} of $u$ and $u$ is a \textit{parent} of $v$ and write $u\prec v$. If $u\prec v$, we call the edge $uv$ \textit{a link from} $u$ \textit{to} $v$.
A \textit{path from} $u$ \textit{to} $v$ is a sequence $w_1\prec w_2\prec\cdots\prec w_n$ where $w_1=u$ and $w_2=v$. We call $n-1$ the
\textit{length} of this path. In general, there may not be a path from $u$ to $v$.

\begin{thm}       
\label{thm42}
{\rm (i)}\enspace There exists a path from $u$ to $v$ if and only if $u<v$.
{\rm (ii)}\enspace $u\prec v$ if and only if $v^0=u^0+1$ and $\doubleab{v-u}_4=0$ or $1$.
{\rm (iii)}\enspace If $u<v$, then any two paths from $u$ to $v$ have length $v^0-u^0$.
\end{thm}
\begin{proof}
(i)\enspace If there exists a path from $u$ to $v$, then $u<v$ by transitivity. Conversely, suppose $u<v$. If $u\prec v$, the $u,v$ form a path and we are finished. If $u\not\prec v$, then there is a $w\in\cscript _0^+$ such that $u<w<v$. If $u\prec w\prec v$, then $u,w,v$ form a path. Otherwise, there is a $w_1\in\cscript _0^+$ such that $u<w_1<w<v$ or $u<w<w_1<v$. Since there are only a finite number of vertices between $u$ and $v$, this process must eventually end and we obtain a path from $u$ to $v$.
(ii)\enspace In Lemma~\ref{lem21} we constructed the 12 unit vectors in $\sscript _3$. It will be convenient to label them by
$\undere _1=e,\undere _2=f,\undere _3=g,\undere _4=e-f,\undere _5=e-g,\undere _6=f-g,\undere _7=-e,\undere _8=-f,\ldots ,\undere _{12}=g-f$
We can construct $\cscript _0^+$ as follows. The universe histories $\overs ^0$ and $\overs ^1$ are given by $\overs ^0=\brac{0}$,
\begin{equation*}
\overs ^1=\brac{0,(1,\underzero ),(1,\undere _1),(1,\undere _2),\ldots ,(1,\undere _{12})}
\end{equation*}
To obtain $\overs ^2$ we add the vectors $(1,\underzero),(1,\undere _1),\ldots ,(1,\undere _{12})$ to those in $\overs ^1$
\begin{align*}
\overs ^2&=\left\{0,(1,\underzero ),(1,\undere _1),
     \ldots ,(1,\undere _{12}),(2,\underzero),(2,\undere _1),(2,2\undere _1),(2,\undere _1+\undere _2),\right.\\
     &\quad\left.\dots ,(2,\undere _{11}+\undere _{12}\right\}
\end{align*}
Strictly speaking, there are repeats in this list for $\overs ^2$ which should be eliminated. Continue this process to obtain the universe histories
$\overs ^t$ for all $t\in\positive$. If $v^0=u^0+1$ and $\doubleab{v-u}_4=0$ or $1$, then $u<v$. Suppose there is a $w\in\cscript _0^+$ such that $u<w<v$. Then $u^0<w^0<v^0$ which is impossible so $u\prec v$. Conversely, suppose $u\prec v$. By our previous construction of
$\cscript _0^+$ we have that
\begin{equation*}
(v^0,\underv )=(u^0,\underu )+(1,\underzero )=(u^0+1,\underu )
\end{equation*}
or for some $j=1,2,\ldots ,12$
\begin{equation*}
(v^0,\underv )=(u^0,\underu )+(1,\undere _j)=(u^0+1,u+\undere _j)
\end{equation*}
In either case, $v^0=u^0+1$ and in the first case
\begin{equation*}
\doubleab{v-u}_4^2=1-0=1
\end{equation*}
while in the second case
\begin{equation*}
\doubleab{v-u}_4^2=1-\doubleab{\undere _j}_3^2=1-1=0
\end{equation*}
(iii)\enspace If $u<v$ then $u\in\sscript ^t$ and $v\in\sscript ^{t'}$ for $t<t'$. From (i) there is a path from $u$ to $v$ and by the construction in (ii), it is clear that every path from $u$ to $v$ has length $t'-t$.
\end{proof}

Applying Theorem~\ref{thm42}(iii), all paths from $0$ to $v$ have length $v^0$. We then say that the \textit{height} of $v$ is $v^0$. We conclude that $\sscript ^t$ is precisely the set of vertices with height $t$ and call $\sscript ^t$ the $t$-\textit{th shell} in $\overs ^{t'}$ for $t\le t'$. Also, notice from the construction in Theorem~\ref{thm42}(ii) that every $u\in\cscript _0^+$ has precisely 13 children. We do not know the number of parents a vertex has and it would be interesting to find out. This would be useful in finding the cardinalities of $\sscript ^t$ and $\overs ^t$ which are also unknown.

If a causet $P$ has the property that whenever $u,v\in P$ with $u<v$, then any two paths from $u$ to $v$ have the same length, we call $P$
\textit{weakly covariant}. For a general causet $P$ the \textit{height} $h(v)$ of $v\in P$ is the length of the longest path terminating at $v$.
We say that $P$ is \textit{covariant} if $h(u)<h(v)$ implies that $u<v$ \cite{gud13,gud151,gud152}. We have shown in previous works that a covariant causet has a natural weak metric and notions of curvature and geodesics. We have employed these concepts to develop a discrete quantum gravity \cite{gud13,gud151,gud152}.

\begin{lem}       
\label{lem43}
A covariant causet is weakly covariant.
\end{lem}
\begin{proof}
Let $P$ be a covariant causet and let $v\in P$ with  $v>0$. Let $w_1\prec w_2\prec\cdots\prec w_n$ be a path from $0$ to $v$. It follows from covariance that $h(w_j)=j$ so that $h(v)=h(w_n)=n$. We conclude that any path from $0$ to $v$ has length $n$. Moreover, if $u<v$, then any path from $u$ to $v$ has length $h(v)-h(u)$.
\end{proof}

Applying Theorem~\ref{thm42}(iii) we conclude that the causets $\overs ^t$, $t>2$, are weakly covariant and it is easy to check that they are not covariant so the converse of Lemma~\ref{lem43} does not hold. However, $\overs ^t$ for $t$ up to some limit $t_1$ (possibly about 300) describes an inflationary period in which the universe histories are essentially flat and gravity has not yet taken effect. After time $t_1$ most of the spacetime cells have been formed and the system goes into a multiverse period. During this period, there are myriads of possible universes which are expanding much more slowly. The universes develop curvatures which are the cause of gravity. In the multiverse period, the universe histories become covariant and this is employed to describe curvature \cite{gud151}. However, the curvature is local and the universes are essentially flat because most of the spacetime cells have already been formed.

We assume that if a particle moves from $u$ to $v$, $u,v\in\cscript _0^+$ with $u<v$, then it traverses a path
$w_1\prec w_2\prec\cdots\prec w_n$, from $u$ to $v$ where $w_1=u$, $w_n=v$. The length of the path is $n-1=v^0-u^0$. By
Theorem~\ref{thm42}(ii) the time interval between $w_j$ and $w_{j+1}$ is 1 unit and the space distance moved is either 0 or 1 unit. This tells us that the \textit{instantaneous speed} of a particle is either 0 or 1. Thus, a particle that is not motionless can move at only one speed. If we are indeed using Planck units, this speed is the speed of light $c$ in a vacuum. This fundamental principle is the reason that $c$ is the speed limit for physical signals. What about objects that we know move slower than $c$? They appear to be moving slower that $c$ because we are actually measuring average speeds in our observations. If a particle propagates from $u$ to $v$, we define its \textit{average speed} to be
$s_a=\doubleab{\underv -\underu}_3/(v^0-u^0)$. Of course, $0\le s_a\le 1$ so $c$ is still the speed limit for average speed. Notice that $s_a$ can have various values. For example, if $u=0$ and $v=3d+3e$, then $s_a=1$, while if $u=0$ and $v=3d+2e-f$ then $s_a=1/\sqrt{3}$.

It is interesting to find the possible average speeds. Of course, there are only a finite number of possibilities for $u,v\in\overs ^t$ with $t$ fixed. For simplicity, suppose a particle propagates from $0$ to $v$ so that $s_a=\doubleab{\underv}_3/v^0$. For $v^0=1$, the only possibilities are $s_a=0,1$. For $v^0=2$, the possibilities are $s_a=0,1/2,\sqrt{2}/2,\sqrt{3}/2,1$. For $v^0=3$, we have
\begin{equation*}
s_a=0,1/3,\sqrt{2}/3,\sqrt{3}/3,\ldots ,\sqrt{8}/\sqrt{3},1
\end{equation*}
But now the pattern ends and there are gaps. For $v^0=4$ we have
\begin{equation*}
s_a=0,\sqrt{3}/4,\sqrt{4}/4,\sqrt{5}/4,\ldots ,\sqrt{15}/4,1
\end{equation*}
For $v^0=5$ we obtain
\begin{equation*}
s_a=0,1/5,\sqrt{2}/5,2/5,\sqrt{6}/5,\sqrt{7}/5,3/5,\sqrt{12}/5,\sqrt{13}/5,\ldots ,\sqrt{24}/5,1
\end{equation*}
We do not know a general formula for possible $s_a$ corresponding to $v^0$.

A \textit{symmetry} on $\sscript _4$ is a linear bijection $T\colon\sscript _4\to\sscript _4$ that preserves the norm $\doubleab{\ctimes}^4$ and has unit determinant. A symmetry $T$ is a \textit{boost} if $Td\ne d$. Notice, if $Td=d$ then $T\sscript _3=\sscript_3$ because
\begin{equation*}
0=\elbows{d,e}=\elbows{Td,Te}=\elbows{d,Te}
\end{equation*}
so $Te\in\sscript _3$ and similarly, $Tf,Tg\in\sscript _3$. Symmetries that are not boosts have the form $1\oplus Z$, $Z\in\gscript (\sscript _3)$. The question that now presents itself is: Are there any boosts? The simplest examples in classical special relativity are the \textit{basic} boosts of the form \cite{sw64,vel94}
\begin{equation*}
T=\begin{bmatrix}\noalign{\smallskip}
1/\sqrt{1-v^2}&v/\sqrt{1-v^2}&0&0\\\noalign{\smallskip}
v/\sqrt{1-v^2}&1\sqrt{1-v^2}&0&0\\\noalign{\smallskip}
0&0&1&0\\\noalign{\smallskip}0&0&0&1\\\noalign{\smallskip}\end{bmatrix}
\end{equation*}
The transformation $T$ describes a coordinate system moving at the constant velocity $v$ along the $x^1$ axis. But in our situation, we only have two values for an instantaneous velocity, $v=0$ or $1$. In the first case, $T=I$ which is not a boost, while the second case is impossible. Even if other values of $v$ are allowed (average velocities), it is easy to check that $T$ does not leave $\sscript _4$ invariant.

\begin{thm}       
\label{thm44}
There are no boosts.
\end{thm}
\begin{proof}
Suppose $T\colon\sscript _4\to\sscript _4$ is a boost and
\begin{align*}
Td&=t_0d+t_1e+t_2f+t_3g\\
Te&=s_0d+s_1e+s_2f+s_3g
\end{align*}
Then
\begin{equation*}
1=\doubleab{d}_4^2=\doubleab{Td}_4^2=t_0^2-\tfrac{1}{2}\sqbrac{(t_1+t_2)^2+(t_1+t_3)^2+(t_2+t_3)^2}
\end{equation*}
Hence,
\begin{equation}         
\label{eq42}
(t_1+t_2)^2+(t_1+t_3)^2+(t_2+t_3)^2+2=2t_0^2
\end{equation}
The only nontrivial solutions to \eqref{eq42} are
\begin{align}         
\label{eq43}
t_0&=2,\ t_1=t_2=1,\ t_3=0\\
\label{eq44}
t_0&=3,\ t_1=t_2=2,\ t_3=-2
\end{align}
and similar ones. Notice that trivial solutions like $t_0=1$, $t_1=t_2=t_3=0$ do not give boosts because then $Td=d$. We also have
\begin{equation*}
-1=\doubleab{e}_4^2=\doubleab{Te}_4^2=s_0^2-\tfrac{1}{2}\sqbrac{(s_1+s_2)^2+(s_1+s_3)^2+(s_2+s_3)^2}
\end{equation*}
Hence,
\begin{equation}         
\label{eq45}
(s_1+s_2)^2+(s_1+s_3)^2+(s_2+s_3)^2=2s_0^2+2
\end{equation}
The only nontrivial solutions of \eqref{eq45} are
\begin{align}         
\label{eq46}
s_0&=1,\ s_1=s_2=1,\ s_3=-1\\
\label{eq47}
s_0&=2,\ s_1=2,\ s_2=1,\ s_3=-1
\end{align}
and similar ones. We also must satisfy
\begin{equation}         
\label{eq48}
0=\elbows{d,e}=\elbows{Td,Te}=t_0s_0-\tfrac{1}{2}(t_1,s_1+t_2s_2+t_3s_3)
\end{equation}
The only combination of cases \eqref{eq43}, \eqref{eq44} with cases \eqref{eq46}, \eqref{eq47} that satisfy \eqref{eq48} are \eqref{eq44} and
\eqref{eq46}. Thus,
\begin{equation*}
Te=d+e+f-g
\end{equation*}
But the same reasoning shows that $Tf=d+e+f-g$. Hence, $Te=Tf$ which contradicts the bijectivity of $T$.
\end{proof}

We conclude from Theorem~\ref{thm44} that all symmetries of $\sscript _4$ have the form $1\oplus Z$, $Z\in\gscript (\sscript _3)$. We denote this group by $\gscript (\sscript _4)$.

\section{Discrete Quantum Field Theory} 
As discussed in the previous section, $\sscript _4$ is the discrete spacetime background for expanding universe histories during an inflationary period. Most of the spacetime cells are formed and the universe is essentially flat during this period. Curvatures and gravity do not emerge until later. We now discuss discrete quantum field theory. As usual, this theory only employs special relativity and the general relativity of gravity is neglected. This assumption is well-founded because gravitation is extremely weak compared to electromagnetic and nuclear forces. For this reason we shall take $\sscript _4$ as the underlying spacetime for our quantum field theory.

Let $\sscripthat _4$ be a copy of $\sscript _4$ whose elements are labeled by
\begin{equation*}
p=(p^0,\underp )=(p^0,p^1,p^2,p^3)
\end{equation*}
in Cartesian coordinates. We think of $\sscripthat _4$ as being dual to $\sscript _4$ with indefinite inner product
\begin{equation*}
p\ctimes x=p^0x^0-p^1x^1-p^2x^2-p^3x^3
\end{equation*}
where $x\in\sscript _4$ is given in Cartesian coordinates. We call $p^0$ the \textit{total energy}, $\underp$ the \textit{momentum} and
\begin{equation*}
m^2=\doubleab{p}_4^2=(p^0)^2-||\underp ||_3^2=(p^0)^2-(p^1)^2-(p^2)^2-(p^3)^2
\end{equation*}
the \textit{mass squared}. We only consider $m^2$ when $m^2\ge 0$. We see that $m^2$ can only have integer values so mass is discrete in this theory. We have not found a formula for the possible values of $m^2$ as a function of $p^0$ and this would be of interest to know. We have computed the values of $||\underp ||_3^2$ up to 49 which are the following:
\begin{equation*}
||\underp ||_3^2=0,1,2,\ldots ,13,16,18,19,21,23,\ldots ,28,31,33,35,\ldots ,39,43,49
\end{equation*}
We then obtain the following mass squared values as a function of $p^0$
\bigskip

\noindent
\begin{tabular}{c|c}
$p^0$&$m^2$\hfill\hfill\\
\hline
0&0\hfill\hfill\\ 
1&0,\ 1\hfill\hfill\\
2&0,\ 1,\ 2,\ 3,\ 4\hfill\hfill\\
3&0,\ 1,\ 2,\ $\ldots$,\ 9\hfill\hfill\\
4&0,\ 3,\ 4,\ $\ldots$,\ 16\hfill\hfill\\
5&0,\ 1,\ 2,\ 4,\ 6,\ 7,\ 9,\ 12,\ 13,\ $\ldots$,\ 25\hfill\hfill\\
6&0,\ 1,\ 3,\ 5,\ 8,\ 9,\ $\ldots$,\ 13,\ 15,\ 17,\ 18,\ 20,\ 23,\ 24,\ $\ldots$,\ 36\hfill\hfill\\
7&0,\ 6,\ 10,\ $\ldots$,\ 14,\ 16,\ 18,\ 21,\ $\ldots$,\ 26,\ 28,\ 30,\ 31,\ 33,\ 36,\ $\ldots$,\ 49\\
\end{tabular}
\bigskip

We can identify $\gscript (\sscript _4)$ with $\gscript (\sscripthat _4)$ and when we write $Z\in\gscript (\sscript _4)$ we mean
$Z\in\gscript (\sscript _4)$ or $Z\in\gscript (\sscripthat _4)$ whichever is applicable. The set $\sscript _4\times\gscript (\sscript _4)$ becomes a group with product
\begin{equation*}
(y,Y)(z,Z)=(y+Yz,YZ)
\end{equation*}
We think of $\sscript _4\times\gscript (\sscript _4)$ as a discrete Poincare group. For $m^2\ge 0$, the \textit{mass hyperboloid} is the set
\begin{equation*}
\Gamma _m=\brac{p\in\sscripthat _4\colon\doubleab{p}_4^2=m^2}
\end{equation*}
The basic Hilbert space for this theory is
\begin{equation*}
H_m=L_2(\Gamma _m)=\brac{f\colon\Gamma _m\to\complex\colon\sum _{p\in\Gamma _m}\ab{f(p)}^2<\infty}
\end{equation*}
with the usual inner product
\begin{equation*}
\elbows{f,g}=\sum _{p\in\Gamma _m}\overline{f(p)}g(p)
\end{equation*}

Define the representation $V$ of $\sscript _4\times\gscript (\sscript _4)$ on $H_m$ by
\begin{equation*}
\sqbrac{V(y,Y)f}(p)=e^{ip\cdot y}f(Y^{-1}p)
\end{equation*}
To show that $V$ gives a representation, we have
\begin{align*}
\sqbrac{V(y,Y)V(z,Z)f}(p)&=e^{ip\cdot y}\sqbrac{V(z,Z)f}(Y^{-1}p)\\
  &=e^{ip\cdot y}e^{iY^{-1}p\cdot z}f(Z^{-1}Y^{-1}p)\\
  &=e^{ip\cdot (y+Yz)}f\paren{(YZ)^{-1}p}\\
  &=\sqbrac{V(y+Yz,YZ)f}(p)\\
  &=\sqbrac{V\paren{(y,Y)(z,Z)}f}(p)
\end{align*}
where the third equality follows from
\begin{equation*}
Y^{-1}p\cdot z=YY^{-1}p\cdot Yz=p\cdot Yz
\end{equation*}

Notice that $Z\in\gscript (\sscript _4)$ leaves $\Gamma _m$ invariant because
\begin{equation*}
||Zp||_4^2=\doubleab{p}_4^2=m^2
\end{equation*}
It follows that $V$ is a unitary representation. Indeed,
\begin{align*}
\elbows{V(z,Z)f,V(z,Z)g}&=\sum _{p\in\Gamma _m}\overline{\sqbrac{V(z,Z)f}(p)}\sqbrac{V(z,Z)g}(p)\\
  &=\sum _{p\in\Gamma _m}e^{-ip\cdot z}\overline{f(Z^{-1}p)}e^{ip\cdot z}g(Z^{-1}p)\\
  &=\sum _{p\in\Gamma _m}\overline{f}(p)g(p)=\elbows{f,g}
\end{align*}

We have the four self-adjoint operators $P^j$, $j=0,1,2,3$, on $H_m$ given by $(P^jf)(p)=p^jf(p)$. Since
\begin{align*}
(P^0)^2f(p)&=(p^0)^2f(p)=\sqbrac{(p^1)^2+(p^2)^2+(p^3)^2+m^2}f(p)\\
  &=\sqbrac{(P^1)^2+(P^2)^2+(P^3)^2+m^2}f(p)
\end{align*}
we conclude that
\begin{equation*}
(P^0)^2=(P^1)^2+(P^2)^2+(P^3)^2+m^2I
\end{equation*}
The eigenvectors of $P^2$ are the characteristic functions $\chi _p$, $p\in\Gamma _m$ with eigenvalues $p^j$, $j=0,1,2,3$, on the mass hyperboloid.

The space $H_m$ that we have considered until now is the \textit{scalar} (or \textit{spin}-0) \textit{mass} $m$ Hilbert space. The \textit{vector}
(or \textit{spin}-1) \textit{mass} $m$ Hilbert space is $H_m\otimes\complex ^3$ with the usual inner product. The representation $V_1$ of
$\sscript _4\times\gscript (\sscript _4)$ on $H_m\otimes\complex ^3$ is given by
\begin{equation*}
V_1(y,Y)f(p)\otimes v=e^{ip\cdot y}f(Y^{-1}p)\otimes U(Y^{-1})v
\end{equation*}
As before, $V_1$ is a unitary representation. The \textit{spin}-$1/2$ \textit{mass} $m$ Hilbert space is $H_m\otimes\complex ^2$ with the usual inner product. The unitary representation $V_{1/2}$ of $\sscript _4\times\gscript (\sscript _4)$ on $H_m\otimes\complex ^2$ is given by
\begin{equation*}
V_{1/2}(y,Y)f(p)\otimes v=e^{ip\cdot y}f(Y^{-1}p)\otimes\rscript (Y^{-1})v
\end{equation*}
We can continue to form the \textit{spin}-$n/2$ \textit{mass} $m$ Hilbert space
\begin{equation*}
H_m\otimes\complex ^2\otimes\cdots\otimes\complex ^2
\end{equation*}
where there are $n$ factors of $\complex ^2$ with the unitary representation
\begin{equation*}
V_{n/2}(y,Y)f(p)\otimes v_1\otimes\cdots\otimes v_n
  =e^{ip\cdot y}f(Y^{-1}p)\otimes\rscript (Y^{-1})v_1\otimes\cdots\otimes\rscript (Y^{-1})v_n
\end{equation*}

We finally develop a discrete quantum field theory. For simplicity, we shall only consider scalar fields and the extension to nonzero spins is fairly straightforward. We first form the Hilbert space $\hscript _m=\oplus_{n=0}^\infty H_m^n$ where $H_m^0=\complex$ and for $n\in\positive$,
$H_m^n$ is the symmetric tensor product
\begin{equation*}
H_m^n=H_m\circleds H_m\circleds\cdots\circleds H_m\quad (n\hbox{ factors})
\end{equation*}
We can consider $f\in H_m^n$ as a symmetric function of $n$ variables $f(p_1,\ldots ,p_n)$, $p_j\in\Gamma _m$, $j=1,2,\ldots ,n$. For
$x\in\sscript _4$ we construct the \textit{field operators} $\phi (x)\colon\hscript _m\to\hscript _m$ by defining $\phi (x)\colon H_m^{n+1}\to H_m^n$ as follows. For $f\in\hscript _m$ we have
\begin{equation}        
\label{eq51}
f=f^{(0)}\oplus f^{(1)}\oplus f^{(2)}\oplus\cdots
\end{equation}
where $f^j\in H_m^j$, $j=0,1,2,\ldots\,$. We define $\phi (x)f^{(0)}=0$ and for $n\in\positive$ we define \cite{sw64,vel94}
\begin{equation*}
\sqbrac{\phi (x)f^{(n+1)}}(p_1,p_2,\ldots ,p_n)=\sqrt{n+1}\sum _{p\in\Gamma _m}e^{-ip\cdot x}f^{(n+1)}(p,p_1,\ldots ,p_n)
\end{equation*}
We also construct the field operators $\psi (x)\colon\hscript _m\to\hscript _m$ by defining $\psi (x)\colon H_m^n\to H_m^{n+1}$ as follows
\begin{equation*}
\sqbrac{\psi (x)f^{(n)}}(p_1,p_2,\ldots ,p_{n+1})=\frac{1}{\sqrt{n+1}}\sum _{j=1}^{n+1}e^{ip_j\cdot x}f^{(n)}(p_1,\ldots ,\phat _j,\ldots ,p_{n+1})
\end{equation*}
where $\phat _j$ means omit $p_j$. We interpret $\phi (x)$ as an operator that annihilates a mass $m$ particle at the spacetime point $x$ while
$\psi (x)$ is an operator that creates a mass $m$ particle at the spacetime point $x$. We obtain a unitary representation of
$\sscript _4\times\gscript (\sscript _4)$ on $\hscript _m$ by defining
\begin{equation*}
\sqbrac{V(y,Y)f^{(n)}}(p_1,\ldots ,p_n)=e^{i\sum _{j=1}^np_j\cdot y}f^{(n)}(Y^{-1}p_1,\ldots ,Y^{-1}p_n)
\end{equation*}
on $H_m^n$ and extending $V(y,Y)$ to $\hscript _m$ by applying \eqref{eq51}.

\begin{thm}       
\label{thm51}
For every $x\in\sscript _4$, $\psi (x)=\phi (x)^*$.
\end{thm}
\begin{proof}
Let $f,g\in\hscript _m$. Since $f$ and $g$ are symmetric in their variables and letting $p=p_{n+1}$ in the fourth equality, we obtain
\begin{align*}
\elbows{\phi (x)f^{(n+1)},g^{(n)}}&=\elbows{\sqrt{n+1}\sum _{p\in\Gamma _m}e^{-ip\cdot x}f^{(n+1)}(p,p_1,\ldots ,p_n),g^{(n)}(p_1,\ldots ,p_n)}\\
  &=\sqrt{n+1}\sum _{p_1,\ldots ,p_n}\sum _pe^{ip\cdot x}\overline{f^{(n+1)}(p,p_1,\ldots ,p_n)}g^{(n)}(p_1,\ldots ,p_n)\\
  &=\frac{1}{\sqrt{n+1}}\left\{\left[\sum _{p,p_1,\ldots ,p_n}e^{ip\cdot x}\overline{f^{(n+1)}(p,p_1,\ldots ,p_n)}\right.\right.\\
  &\quad +\sum _{p,p_1,\ldots ,p_n}e^{ip_1\cdot x}\overline{f^{(n+1)}(p,p_1,\ldots ,p_n)}+\cdots\\
  &\quad\left.\left. +\sum _{p,p_1,\ldots ,p_n}e^{ip_n\cdot x}\overline{f^{(n+1)}(p,p_1,\ldots ,p_n)}\right]g^{(n)}(p_1,\ldots ,p_n)\right\}\\
    &=\frac{1}{\sqrt{n+1}}\sum _{p_1,\ldots ,p_n}\overline{f^{(n+1)}(p_1,\ldots ,p_{n+1})}\\
    &\quad\quad\sum _{j=1}^{n+1}e^{ip_j\cdot x}g^{(n)}(p_1,\ldots ,\phat _j,\ldots ,p_{n+1})\\
    &=\frac{1}{\sqrt{n+1}}\elbows{f^{(n+1)}(p_1,\ldots ,p_{n_1}),
    \sum _{j=1}^{n+1}e^{ip_j\cdot x}g^{(n)}(p_1,\ldots ,\phat _j,\ldots ,p_{n+1})}\\
    &=\elbows{f^{(n+1)},\psi (x)g^{(n)}}
\end{align*}
We conclude that
\begin{align*}
\elbows{\phi (x)f,g}&=\sum _{n=0}^\infty\elbows{\phi(x)f^{(n+1)},g^{(n)}}
  =\sum _{n=0}^\infty\elbows{f^{(n+1)},\psi (x)g^{(n)}}\\
  &=\elbows{f,\psi (x)g}
\end{align*}
Hence, $\psi (x)=\phi (x)^*$.
\end{proof}

It is not hard to show that the commutator $\sqbrac{\phi (x),\phi (y)}=0$ for all $x,y\in\sscript _4$. Similarly, $\sqbrac{\psi (x),\psi (y)}=0$ for all
$x,y\in\sscript _4$.

\begin{thm}       
\label{thm52}
For all $x,y\in\sscript _4$ we have that
\begin{equation*}
\sqbrac{\phi (x),\psi (y)}=\sum _{p\in\Gamma _m}e^{ip\cdot (y-x)}I
\end{equation*}
\end{thm}
\begin{proof}
For $f^{(n)}\in H_m^n$ we have
\begin{align*}
\sqbrac{\phi (x)\psi (y)f^{(n)}}&(p_1,\ldots ,p_n)=\sqrt{n+1}\sum _{p\in\Gamma _m}e^{-ip\cdot x}\sqbrac{\psi (y)f^{(n)}}(p,p_1,\ldots ,p_n)\\
  &=\sum _{p\in\Gamma _m}e^{-ip\cdot x}\left[e^{ip\cdot y}f^{(n)}(p_1,\ldots ,p_n)+e^{ip_1\cdot y}f^{(n)}(p,p_1,\ldots ,p_n)\right.\\
  &\quad\left.+\cdots +e^{ip_n\cdot y}f^{(n)}(p,p_1,\ldots ,p_n)\right]
\end{align*}
On the other hand,
\begin{align*}
\sqbrac{\psi (y)\phi (x)f^{(n)}}&(p_1,\ldots ,p_n)
  =\frac{1}{\sqrt{n}}\sum _{j=1}^ne^{ip_j\cdot y}\sqbrac{\phi (x)f^{(n)}}(p_1,\ldots ,\phat _j,\ldots ,p_n)\\
 &=\sum _{j=1}^n\sum _{p\in\Gamma _m}e^{-ip\cdot x}e^{ip_j\cdot y}f^{(n)}(p,p_1,\ldots ,\phat _j,\ldots ,p_n)\\
 &=\sum _{p\in\Gamma _m}e^{-ip\cdot x}\left[e^{ip_1\cdot y}f^{(n)}(p,p_2,\ldots ,p_n)+e^{ip_2\cdot y}f^{(n)}(p,p_1,p_3,\ldots ,p_n)\right.\\
 &\quad\quad\left.+\cdots +e^{ip_n\cdot y}f^{(n)}(p_,p_1,p_2,\ldots ,p_{n+1})\right]
\end{align*}
Therefore,
\begin{equation*}
\sqbrac{\phi (x),\psi (y)}f^{(n)}=\sum _{p\in\Gamma _m}e^{ip\cdot (y-x)}f^{(n)}
\end{equation*}
It follows that
\begin{equation*}
\sqbrac{\phi (x),\psi (y)}f=\sum _{p\in\Gamma _m}e^{ip\cdot (y-x)}f
\end{equation*}
for all $f\in\hscript _m$
\end{proof}

We also define the self-adjoint field operators $\xi (x)=\phi (x)+\psi(x)$ for all $x\in\sscript _4$.
\begin{cor}       
\label{cor53}
For all $x,y\in\sscript _4$ we have that
\begin{equation*}
\sqbrac{\xi (x),\xi (y)}=2i\sum _{p\in\Gamma _m}\sin p\cdot (y-x)I
\end{equation*}
\end{cor}
\begin{proof}
Since $\sqbrac{\phi (x),\phi (y)}=\sqbrac{\psi (x),\psi (y)}=0$ we have that
\begin{equation*}
\sqbrac{\xi (x),\xi (y)}x=\sqbrac{\phi (x),\psi (y)}+\sqbrac{\psi (x),\phi (y)}
\end{equation*}
Applying Theorem~\ref{thm52} gives
\begin{align*}
\sqbrac{\xi (x),\xi (y)}&=\sum _{p\in\Gamma _m}\sqbrac{e^{ip\cdot (y-x)}-e^{-ip\cdot (y-x)}}I\\
  &2i\sum _{p\in\Gamma _m}\sin p\cdot (y-x)I\qedhere
\end{align*}
\end{proof}

The quantum field theory that we have developed is essentially trivial because the fields are free with no interactions. We view this as just the beginning, and initiate the difficult task of investigating interacting fields in the next section. It is hoped that these initial steps may result in a mathematically rigorous quantum field theory without singularities.

\section{Interacting Quantum Fields} 
This section provides a simple example of interacting quantum fields. If $X(n)$, $n=0,1,2,\ldots$, are operators on a Hilbert space, we define the \textit{difference operator} $\nabla X$ by $\nabla X(n)=X(n+1)-X(n)$. We begin with a simple, but useful, lemma.
\begin{lem}       
\label{lem61}
If $X(n)$ and $A(n)$ are operators on the same Hilbert space, $n=0,1,2,\ldots$, satisfying $\nabla X(n)=A(n)X(n)$, then
\begin{align*}
X(n)&=\sqbrac{I+A(n-1)}\sqbrac{I+A(n-2)}\cdots\sqbrac{I+A(0)}X(0)\\
  &=\left[I+\sum _{j=0}^{n-1}A(j)+\sum\brac{A(j_1)A(j_2)\colon j_1,j_2=0,1,\ldots ,n-1,\ j_2<j_1}\right.\\
  &\quad +\sum\brac{A(j_1)A(j_2)A(j_3)\colon j_1,j_2,j_3=0,1,\ldots ,n-1,\ j_3<j_2<j_1}\\
  &\quad\biggl. +\cdots +A(n-1)A(n-2)\cdots A(1)A(0)\biggr]X(0)
\end{align*}
\end{lem}
\begin{proof}
We have that $X(n+1)-X(n)=A(n)X(n)$ so $X(n+1)=\sqbrac{I+A(n)}X(n)$. Replacing $n+1$ by $n$ we obtain
\begin{align*}
X(n)&=\sqbrac{I+A(n-1)}X(n-1)=\sqbrac{I+A(n-1)}\sqbrac{I+A(n-2)}X(n-2)\\
   &=\cdots =\sqbrac{I+A(n-1)}\sqbrac{I+A(n-2)}\cdots\sqbrac{I+A(0)}X(0)
\end{align*}
The second equality follows by induction.
\end{proof}
For $x\in\sscript _4$ define the field operators $\pi (x)=\phi (x)+\phi (x)^*$ on the Hilbert space $\hscript _m$ as discussed in Section~5. Also, let $\sigma (x)$, $x\in\sscript _4$, be similarly defined field operators on the Hilbert space $\hscript _M$, $M\ge 0$. We think of $\pi (x)$ as describing $\pi$-particles of mass $m$ and $\sigma (x)$ as describing $\sigma$-particles of mass $M$. The interaction between these two types of particles will be described on the tensor product $\hscript =\hscript _m\otimes\hscript _M$. Central roles are played by two operators on $\hscript$. These operators are functions of time $x_0=0,1,2,\ldots$, and are the self-adjoint \textit{interaction Hamiltonian} $H(x_0)$ and the
\textit{scattering operator} $S(x_0)$. We assume that $S(0)=I$ and that $S(x_0)$ satisfies the discrete Schr\"odinger's equation:
\begin{equation*}
-i\nabla _{x_0}S(x_0)=H(x_0)S(x_0)
\end{equation*}
It follows from Lemma~\ref{lem61} that
\begin{align}         
\label{eq61}
S(n)&=I+i\sum _{j=0}^{n-1}H(j)+i^2\sum\brac{H(j_1)H(j_2)\colon j_1,j_2=0,1,2,\ldots ,n-1,\ j_2<j_1}\notag\\
  &\quad +i^3\sum\brac{H(j_1)H(j_2)H(j_3)\colon j_1,j_2,j_3=0,1,2,\ldots ,n-1,\ j_3<j_2<j_1}\notag\\
  &\quad +\cdots +i^nH(n-1)H(n-2)\cdots H(0)
\end{align}
It is interesting that \eqref{eq61} resembles the inclusion-exclusion principle which is useful in probability and combinatorics theory. In the usual continuum theory, \eqref{eq61} has the form of a very complicated integral involving time-ordered products. In \eqref{eq61} the time-ordering is simpler and results in fewer terms. Essentially the only experiments available in quantum field theory are scattering experiments and the main objective in quantum field theory is to calculate the \textit{final scattering operator} $S=\lim _{n\to\infty}S(n)$. Unfortunately, \eqref{eq61} cannot usually be solved in closed form to find $S$. The only thing we can do is use approximations or perturbative techniques. The interaction Hamiltonian is usually given in terms of the \textit{Hamiltonian density} $K(x)$ by
\begin{equation*}
H(x_0)=\sum\brac{K(x_0,x_1,x_2,x_3)\colon\doubleab{x}_4\ge 0}
\end{equation*}
In our particular example, suppose we consider the scattering of two $\pi$-particles interacting with a $\sigma$-particle. We take
$K(x)=g\pi ^2(x)\otimes\sigma (x)$ \cite{vel94}, where $g$ is the \textit{coupling constant}. When $g$ is small, not as many terms are needed for approximations. Assume that the two $\pi$-particles initially have energy-momenta $p$ and $q$ giving rise, after scattering, to two $\pi$-particles with final energy-momenta $p'$ and $q'$. Instead of writing the initial and final states as $\ket{p}\otimes\ket{q}$, $\ket{p'}\otimes\ket{q'}$ we use the simpler notation $\ket{pq}$ and $\ket{p'q'}$, respectively. The probability amplitude for the event of interest is $\bra{p'q'}S\ket{pq}$ and the probability becomes
\begin{equation*}
\ab{\bra{p'q'}S\ket{pq}}^2
\end{equation*}
Applying \eqref{eq61}, the first two terms of $\bra{p'q'}S\ket{pq}$ have the form
\begin{equation*}
\elbows{p'q'\mid pq}+i\sum _{\doubleab{x}_4\ge0}\bra{p'q'}K(x)\ket{pq}
\end{equation*}
Assuming that $\ket{pq}\ne\ket{p'q'}$, the two vectors are orthogonal so the first term is zero. The second term contains one $K(x)$ and hence only one $\sigma (x)$ field. This applied to $\ket{pq}$ gives 0 for the annihilation part or a state of the form $\ket{pq,r}$ for the creation part. Since
$\elbows{p'q'\mid pq,r}=0$ we again obtain 0. Similarly, any term with an odd number of $K(x)$ gives 0. The third term is nonzero and is treated in a similar way. However, there are quite a few possibilities and it appears that the best way to keep track is to employ Feynman diagrams \cite{vel94}. We shall leave the details to later works.

\end{document}